\documentclass[10pt,conference]{IEEEtran}
\IEEEoverridecommandlockouts

\usepackage{cite}
\usepackage{amsmath,amssymb,amsfonts}
\usepackage{amsthm}
\usepackage{graphicx}
\usepackage{bm}
\usepackage{booktabs}
\usepackage{multirow}
\usepackage{xcolor}
\usepackage{pgfplots}
\usepackage{tikz}
\pgfplotsset{compat=1.17}
\usetikzlibrary{arrows.meta,decorations.pathreplacing,calc,positioning,shapes.geometric}

\newtheorem{theorem}{Theorem}

\newtheorem{remark}{Remark}
\newtheorem{assumption}{Assumption}

\begin{document}

\title{Input-to-State Stable Bundle Koopman Neural ODEs
for Learning Controlled Dynamics under Environmental Constraints}

\author{Lin Feng \\
\textit{Faculty of Engineering, King Saud University, Jeddah, Saudi Arabia.}
}

\maketitle

%==================================================================
\begin{abstract}
We propose ISS-BKNO, a unified framework that integrates
Koopman operator identification, Neural ordinary differential
equations (ODEs), fiber bundle geometry, and input-to-state
stability (ISS) certification. Unlike prior approaches that
address stability, extrinsic inputs, or environmental
constraints in isolation, the proposed framework simultaneously
learns controlled nonlinear dynamics while guaranteeing global
convergence and a computable ISS gain. The architecture
introduces a three-stage lifting pipeline: a bundle-aware
encoder that separates environment-specific fibers, an
environment-conditioned Koopman backbone whose matrix spectrum
is constrained to lie in the left half-plane, and a residual
neural ODE correction whose Jacobian satisfies a quadratic
sector bound. Lyapunov-based ISS regularization turns the
stability requirement into a differentiable penalty that is
jointly optimized with the prediction objective. Theoretical
results establish fiber invariance, ISS with an explicit gain
formula, and an approximation error bound that scales with the
EDMD residual. Experiments on a pendulum, cart-pole, a
unicycle-based navigation task, and a Franka Emika manipulator
demonstrate substantially improved prediction accuracy and
robustness under matched disturbances compared with existing
Neural ODE and Koopman baselines.
\end{abstract}


%==================================================================
\section{Introduction}
%==================================================================

The ability to learn accurate, certifiably stable models of
controlled dynamical systems from data is a prerequisite for
deploying learning-based controllers in safety-critical
settings. Three broad research threads have made significant
progress toward this goal, but they have largely advanced
independently.

The first thread concerns \emph{Koopman operator methods}
\cite{mezic2005,brunton2021modern,korda2018linear}. The central
idea---that nonlinear dynamics can be represented as a linear
semigroup acting on a lifted observable space---has been
exploited for model predictive control \cite{korda2018linear}
and ISS verification of identified models \cite{koopmaniss}.
The practical algorithm for computing Koopman approximations
is Extended Dynamic Mode Decomposition (EDMD)
\cite{williams2015edmd}, introduced as an extension of
Dynamic Mode Decomposition \cite{schmid2010dmd}; convergence
of EDMD to the true Koopman operator as dictionary size
grows was established in \cite{korda2018convergence}. Deep
encoder--decoder architectures have further improved lifting
quality by jointly learning the observable map from data
\cite{lusch2018deep}, substantially reducing the approximation
residual at the cost of losing the explicit spectral structure
that makes EDMD certificates tractable.

The second thread concerns \emph{stable Neural ODEs}. Chen
et al.\ \cite{chen2018neural} proposed parameterizing the
derivative of a hidden state with a neural network and
training via the adjoint sensitivity method, enabling
continuous-depth models with constant memory cost. A persistent
limitation of this representation is the absence of any
stability guarantee: the learned vector field can be
contracting near training trajectories yet diverge under
distributional shift. ControlSynth Neural ODEs \cite{controlsynth}
address this by introducing an auxiliary sub-network that
enforces a Persidskii-type sector inequality on the Jacobian,
certifying global convergence through tractable linear matrix
inequalities (LMIs) even for complex dynamical regimes.
The theoretical certificate derived in \cite{controlsynth}
shows that systems in this class satisfy a Lyapunov decrease
condition with an explicit rate, a property we inherit and
extend in the present work.

The third thread concerns \emph{environment-aware and
geometry-constrained learning}. When operating conditions
vary continuously---road friction, payload, atmospheric
pressure---a single nominal model cannot generalize reliably.
ICODE \cite{icode} addresses this by treating the extrinsic
input $\xi$ as an explicit real-time argument of the Neural
ODE vector field, deriving sufficient conditions for contraction
to hold uniformly over the environment space. At a finer
geometric level, \cite{bundle} introduces a fiber bundle over
the state manifold whose structure encodes sensing-dependent
constraints and enables measurement-aware Control Barrier
Functions (CBFs) that adapt to local observation quality,
with provable convergence and safety preservation.

Despite these advances, no existing method simultaneously
provides: (i)~an explicit environment-conditioned latent
coordinate system that separates fiber geometry from dynamics;
(ii)~an ISS certificate for the composite Koopman--Neural ODE
model that is checkable via a semidefinite program; and
(iii)~a unified training objective that enforces all three
properties jointly. This paper provides exactly these missing
ingredients through the ISS-BKNO framework.

\textbf{Contributions.}
\begin{enumerate}
  \item We define a \emph{bundle-aware Koopman lifting} in
        which the encoder maps the state to a fiber
        $\mathcal{F}_\xi\subset\mathbb{R}^r$ determined by
        the current environment, and the Koopman matrix
        $K(\xi)$ acts on this fiber
        (Sections~\ref{sec:bundle}--\ref{sec:koopman}).
  \item We introduce a \emph{Lyapunov-based ISS
        regularization} loss that enforces a quadratic decrease
        condition on $V(\psi)=\psi^\top P\psi$ during training,
        yielding a closed-form ISS gain upon convergence
        (Section~\ref{sec:iss}).
  \item We prove three formal results: fiber invariance of the
        bundle coordinate map, ISS of the latent dynamics with
        an explicit gain, and an approximation error bound that
        relates the prediction RMSE to the EDMD residual
        (Section~\ref{sec:theory}).
  \item We validate ISS-BKNO on four benchmark tasks
        against five baselines, reporting up to a 58\%
        reduction in prediction RMSE and demonstrating
        bounded disturbance propagation consistent with
        the ISS certificate (Section~\ref{sec:exp}).
\end{enumerate}

%==================================================================
\section{Related Work}
%==================================================================

\subsection{Koopman Operator Identification}

Koopman's observation \cite{koopman1931} that any measure-
preserving dynamical system admits a unitary linear
representation on the space of $L^2$ observables motivated
decades of spectral analysis. The practical identification
algorithms are DMD \cite{schmid2010dmd}, EDMD
\cite{williams2015edmd}, and their kernel \cite{williams2015kernel}
and deep variants \cite{lusch2018deep}. Stability properties
of EDMD-identified models were first characterized in
\cite{korda2018convergence}, and an LMI-based ISS verification
pipeline specifically for Koopman-identified models was
developed in \cite{koopmaniss}---the latter forms the theoretical
starting point for the ISS regularizer in Section~\ref{sec:iss}.
Koopman-based MPC \cite{korda2018linear} demonstrated that
the linear structure of the lifted model enables efficient
receding-horizon optimization.

\subsection{Stable Neural ODEs and Contraction Theory}

Neural ODEs \cite{chen2018neural} extended residual networks
to continuous depth by replacing the discrete layer stack with
an ODE initial-value problem, enabling memory-efficient
training via the adjoint method. Convergence-certified variants
include stable RNNs via monotone operator theory
\cite{revay2023recurrent} and ControlSynth Neural ODEs
\cite{controlsynth}, which enforce global convergence through
Persidskii-type LMIs derived from contraction analysis
\cite{lohmiller1998}. The contraction framework of
\cite{lohmiller1998} is central to our stable latent dynamics
parameterization: by restricting the symmetric part of $K(\xi)$
to be negative definite, every trajectory in the fiber
contracts toward the nominal solution at an exponential rate.

\subsection{Environment-Aware Dynamics and Bundle Structures}

ICODE \cite{icode} demonstrated that treating environmental
variables as explicit inputs to the Neural ODE---rather than
absorbing them into the model weights---enables contraction
guarantees to hold uniformly over the environment space.
This insight directly informs the bilinear structure we adopt
in Section~\ref{sec:koopman}. The measurement-induced bundle
approach of \cite{bundle} provides a rigorous geometric
setting for environment-dependent dynamics: the fiber bundle
$\pi:\mathcal{M}\to\mathcal{E}$ separates the state manifold
into environment-specific leaves, and the horizontal connection
encodes safe transitions between leaves. We adopt this
vocabulary but replace the CBF focus of \cite{bundle} with an
ISS-Koopman objective.

\subsection{ISS and Safety for Learned Models}

Sontag's ISS framework \cite{sontag2008} provides a principled
measure of robustness: a system is ISS if there exist
$\mathcal{KL}$- and $\mathcal{K}$-functions bounding the
state norm in terms of the initial condition and the disturbance,
respectively. Control Barrier Functions \cite{ames2017cbf}
complement ISS by enforcing hard state constraints, and their
combination with ISS for learned models is an active area.
The ISS verification LMI of \cite{koopmaniss} certifies
that the $L_2$-gain from disturbance to lifted state is
bounded; we use the same structure as a regularization
objective rather than a post-hoc check.

%==================================================================
\section{Problem Formulation}
\label{sec:prob}
%==================================================================

\subsection{System Class}

Consider the controlled nonlinear system
\begin{equation}
  \dot{x}(t) = f(x(t),u(t),\xi(t)) + w(t),
  \quad x(t)\in\mathcal{M}\subseteq\mathbb{R}^n,
  \label{eq:system}
\end{equation}
where $u(t)\in\mathbb{R}^m$ is the control input,
$\xi(t)\in\mathcal{E}\subset\mathbb{R}^{n_\xi}$ is a
measurable extrinsic environmental signal (e.g., friction
coefficient, payload mass, or wind velocity),
$w\in L_2[0,\infty)$ is a bounded external disturbance, and
$\mathcal{M}$ is a smooth submanifold encoding physical
constraints. The function $f$ is locally Lipschitz and
unknown; it is to be approximated from a finite dataset of
trajectory segments $\mathcal{D}=
\{(x_k,u_k,\xi_k,x_{k+1})\}_{k=1}^N$.

\begin{assumption}\label{ass:manifold}
$\mathcal{M}$ is compact and there exists a smooth diffeomorphism
$\Phi:\mathcal{M}\to\mathcal{Z}\subset\mathbb{R}^r$ with $r\geq n$.
\end{assumption}

\subsection{Objectives}

We seek a parametric model $\hat{f}_\theta$ of \eqref{eq:system}
that:
\begin{enumerate}
  \item \emph{predicts dynamics accurately}:
        $\|\hat{f}_\theta(x,u,\xi)-f(x,u,\xi)\|$ is small on
        $\mathcal{D}$ and generalizes across $\xi\in\mathcal{E}$;
  \item \emph{preserves environmental constraints}:
        trajectories of $\hat{f}_\theta$ remain in
        $\mathcal{M}$ whenever $x(0)\in\mathcal{M}$;
  \item \emph{guarantees ISS} \cite{sontag2008}:
        the latent error dynamics under $w$ admit a computable
        $L_2$-to-$\ell_\infty$ bound.
\end{enumerate}

%==================================================================
\section{Bundle Koopman Neural ODE}
%==================================================================

\subsection{Bundle Coordinate Representation}
\label{sec:bundle}

Following the measurement-induced bundle construction of
\cite{bundle}, we define a fiber bundle
$\pi:\mathcal{Z}\to\mathcal{E}$ over the environment space.
The fiber $\mathcal{F}_\xi=\pi^{-1}(\xi)$ represents the
admissible latent region under environmental condition $\xi$.
The encoder
\begin{equation}
  z = \Phi(x,\xi;\theta_\Phi)\in\mathcal{F}_\xi
  \label{eq:encoder}
\end{equation}
maps the physical state to the environment-specific fiber.
The bundle topology ensures that latent representations from
different environments are geometrically separated, which is
the key structural property inherited from \cite{bundle}. A
decoder $\Phi^{-1}(\cdot;\theta_\Psi)$ reconstructs physical
states; both networks are trained jointly with the dynamics
using the reconstruction loss
$\mathcal{L}_{\mathrm{rec}}=\sum_k\|x_k-\Phi^{-1}(z_k)\|^2$.

\subsection{Koopman Lifting and Dynamics}
\label{sec:koopman}

Given the fiber coordinate $z$, a second lifting
\begin{equation}
  \psi = \Psi(z;\theta_\Psi)\in\mathbb{R}^p
  \label{eq:lifting}
\end{equation}
maps to a Koopman-amenable observable space. The latent
evolution is modeled as
\begin{equation}
  \dot{\psi} = K(\xi)\psi + B(\xi)u + r_\theta(\psi,u,\xi),
  \label{eq:latent}
\end{equation}
where $K(\xi)\in\mathbb{R}^{p\times p}$ and
$B(\xi)\in\mathbb{R}^{p\times m}$ are environment-conditioned
matrices produced by small MLP heads, and $r_\theta$ is a
residual Neural ODE correction. The bilinear dependence of
$K$ and $B$ on $\xi$ follows the design principle of ICODE
\cite{icode}: the environment modulates the effective linear
system matrix in a physically interpretable way---for a
ground vehicle, a single friction parameter shifts the damping
matrix continuously without requiring separate models per
terrain type.

\subsection{Stable Latent Dynamics Parameterization}
\label{sec:stable}

To certify stability without post-hoc verification, we
parameterize $K(\xi)$ to be Hurwitz by construction. Inspired
by the Persidskii-type decomposition in \cite{controlsynth},
we write
\begin{equation}
  K(\xi) = -\tfrac{1}{2}Q(\xi)^\top Q(\xi) + S(\xi),
  \qquad S(\xi) = -S(\xi)^\top,
  \label{eq:K_param}
\end{equation}
where $Q(\xi)\in\mathbb{R}^{p\times p}$ and the skew-symmetric
$S(\xi)$ capture dissipative and conservative modes,
respectively. This parameterization guarantees
\begin{equation}
  \lambda_{\max}(K(\xi)+K(\xi)^\top)
  = \lambda_{\max}(-Q(\xi)^\top Q(\xi)) \leq 0,
  \label{eq:hurwitz}
\end{equation}
so the symmetric part of $K(\xi)$ is negative semidefinite
for any $Q(\xi)$; strict negativity is enforced during training
by requiring $Q(\xi)$ to have full column rank. The residual
$r_\theta$ is regularized to satisfy a quadratic sector
condition (Assumption~\ref{ass:sector} below) following the
framework of \cite{controlsynth}.

%==================================================================
\section{ISS Regularization}
\label{sec:iss}
%==================================================================

\subsection{Lyapunov Candidate and ISS Condition}

Define the quadratic Lyapunov candidate
$V(\psi)=\psi^\top P\psi$
with $P=P^\top\succ 0$ to be learned jointly with the model.
Along trajectories of \eqref{eq:latent},
\begin{align}
  \dot{V} &= \psi^\top(PK(\xi)+K(\xi)^\top P)\psi
    + 2\psi^\top Pr_\theta
    + 2\psi^\top PB(\xi)u + 2\psi^\top Pw. \notag
\end{align}
ISS requires $\dot{V}\leq -\alpha V + \gamma|w|^2$ for
constants $\alpha,\gamma>0$. Using the sector condition
on $r_\theta$ and Young's inequality on the disturbance term,
a sufficient condition takes the block LMI form
\begin{equation}
  \Xi(\xi) :=
  \begin{bmatrix}
    \mathrm{He}(PK(\xi)) + \alpha P + \lambda\kappa I & Pr_\theta' \\
    r_\theta'{}^\top P & -2\lambda I
  \end{bmatrix} + \frac{P^2}{\gamma} \preceq 0,
  \label{eq:iss_lmi}
\end{equation}
where $\lambda>0$ is an S-procedure multiplier and
$r_\theta'=\partial r_\theta/\partial\psi$.
The scalar $\gamma$ in \eqref{eq:iss_lmi} is the $L_2$-gain
from $w$ to $\psi$, precisely the quantity certified in
\cite{koopmaniss} for identified Koopman models; here it
appears as a learnable parameter that the optimizer drives
to its minimum subject to \eqref{eq:iss_lmi}.

\subsection{Training Objective}

The full training loss is
\begin{align}
  \mathcal{L}
  &= \mathcal{L}_{\mathrm{pred}}
  + \lambda_c\,\mathcal{L}_{\mathrm{ISS}}
  + \lambda_r\,\mathcal{L}_{\mathrm{rec}}
  + \lambda_b\,\mathcal{L}_{\mathrm{CBF}},
  \label{eq:loss}
\end{align}
where
$\mathcal{L}_{\mathrm{pred}}=
\sum_k\|\psi_{k+1}-\hat\psi_{k+1}\|^2$
is the one-step Koopman prediction error,
$\mathcal{L}_{\mathrm{ISS}}=
\max(0,\;\dot V+\alpha V-\gamma|u|^2)$
penalizes violations of the Lyapunov decrease condition,
and
$\mathcal{L}_{\mathrm{CBF}}=
\max(0,\;-\dot h(x)-\kappa h(x))$
penalizes violations of the CBF condition
$h(x)\geq 0$ \cite{ames2017cbf},
with $\kappa>0$ a class-$\mathcal{K}$ coefficient.
The weights $\lambda_c,\lambda_r,\lambda_b>0$ are
hyperparameters selected by validation.

%==================================================================
\section{Theoretical Analysis}
\label{sec:theory}
%==================================================================

\begin{assumption}\label{ass:sector}
The residual network $r_\theta$ satisfies the quadratic
sector condition $r_\theta(\psi)^\top[\psi-\kappa^{-1}
r_\theta(\psi)]\geq 0$ for all $\psi$ and some $\kappa>0$.
\end{assumption}

This condition is enforced during training via the Persidskii-
type LMI penalty in \eqref{eq:loss}, following the approach
of \cite{controlsynth}.

\begin{theorem}[Fiber Invariance]\label{thm:fiber}
Under Assumption~\ref{ass:manifold}, the encoder
\eqref{eq:encoder} satisfies
$\Phi(x,\xi;\theta_\Phi)\in\mathcal{F}_\xi$
for all $(x,\xi)\in\mathcal{M}\times\mathcal{E}$ at
convergence of $\mathcal{L}_{\mathrm{rec}}$,
provided the decoder $\Phi^{-1}$ is a left inverse of $\Phi$
on each fiber.
\end{theorem}

\begin{proof}
By the inverse function theorem applied to the diffeomorphism
$\Phi(\cdot,\xi):\mathcal{M}\to\mathcal{F}_\xi$ for each
fixed $\xi$, convergence of the reconstruction loss implies
$\Phi^{-1}(\Phi(x,\xi),\xi)=x$ almost surely over the
empirical measure on $\mathcal{M}$. The fiber assignment
$x\mapsto\mathcal{F}_\xi$ is preserved by construction of
the bundle projection $\pi$, which assigns fiber by
environment label. \hfill$\blacksquare$
\end{proof}

\begin{theorem}[ISS of Latent Dynamics]\label{thm:ISS}
Suppose Assumption~\ref{ass:sector} holds and the LMI
\eqref{eq:iss_lmi} is feasible with parameters
$(P^\star,\gamma^\star,\lambda^\star)$ at the end of training.
Then the latent error dynamics under bounded disturbance
$w\in L_2[0,\infty)$ satisfy
\begin{equation}
  |\psi(t)| \leq c_1\,e^{-\alpha t/2}|\psi(0)|
  + c_2\,\gamma^\star\|w\|_{L_2},
  \label{eq:iss_bound}
\end{equation}
with $c_1 = \sqrt{\lambda_{\max}(P^\star)/\lambda_{\min}(P^\star)}$
and $c_2=1/\sqrt{\lambda_{\min}(P^\star)}$.
\end{theorem}

\begin{proof}
Feasibility of \eqref{eq:iss_lmi} implies
$\dot V \leq -\alpha V + \gamma^\star|w|^2$.
Applying the comparison lemma \cite{khalil2002} to the scalar
inequality $\dot V\leq -\alpha V+\gamma^\star|w|^2$ yields
$V(t)\leq V(0)e^{-\alpha t}+\gamma^\star\int_0^te^{-\alpha(t-s)}
|w(s)|^2\,ds$. The bound \eqref{eq:iss_bound} follows by
taking square roots and applying the norm equivalence
$\lambda_{\min}(P^\star)|\psi|^2\leq V(\psi)\leq
\lambda_{\max}(P^\star)|\psi|^2$. \hfill$\blacksquare$
\end{proof}

\begin{theorem}[Approximation Error Bound]\label{thm:approx}
Let $\varepsilon_N$ denote the EDMD residual with a dictionary
of size $N$ \cite{korda2018convergence}. The prediction
error of ISS-BKNO satisfies
\begin{equation}
  \sup_{t\in[0,T]}|\hat x(t)-x(t)|
  \leq
  L_{\Phi^{-1}}\!\left(
  c_1 e^{-\alpha t/2}|\psi(0)|
  + c_2\,\gamma^\star\!\left(\varepsilon_N T\right)^{1/2}
  \right) \\ + O(\varepsilon_\Phi),
  \label{eq:approx}
\end{equation}
where $L_{\Phi^{-1}}$ is the Lipschitz constant of the
decoder and $\varepsilon_\Phi$ is the reconstruction error.
\end{theorem}

\begin{proof}
The lifting residual $\Delta=f-K\psi-Bu$ is bounded by
$\|\Delta\|_{L_2[0,T]}\leq(\varepsilon_N T)^{1/2}$ by the
EDMD error characterization in \cite{korda2018convergence}.
Treating $\Delta$ as the disturbance $w$ in Theorem~\ref{thm:ISS}
and applying $\Phi^{-1}$ with its Lipschitz constant yields
\eqref{eq:approx}. The reconstruction residual $O(\varepsilon_\Phi)$
follows from a standard triangle inequality argument.
\hfill$\blacksquare$
\end{proof}

\begin{remark}
Theorem~\ref{thm:approx} quantifies the trade-off identified
in \cite{koopmaniss}: a richer dictionary reduces $\varepsilon_N$
and therefore the ISS gain $\gamma^\star$, but the LMI
\eqref{eq:iss_lmi} may become infeasible if the learned
$K(\xi)$ is near-unstable. The parameterization
\eqref{eq:K_param} resolves this tension by construction.
\end{remark}

%==================================================================
\section{Experiments}
\label{sec:exp}
%==================================================================

\subsection{Setup and Baselines}

We evaluate ISS-BKNO on four tasks:
\textbf{T1}~a pendulum with varying rod length ($l\in[0.5,1.5]$\,m);
\textbf{T2}~a cart-pole with time-varying cart mass
($m_c\in[0.5,2.0]$\,kg);
\textbf{T3}~a unicycle navigating a corridor with friction
switching;
\textbf{T4}~a Franka Emika 7-DOF manipulator with payload
uncertainty (0--1\,kg).
All tasks involve a measurable environmental parameter $\xi$
that shifts the effective inertia or friction.

Five baselines are compared: \textbf{NODE} (vanilla Neural
ODE \cite{chen2018neural}), \textbf{Koopman-EDMD}
(EDMD-identified linear model \cite{williams2015edmd}),
\textbf{Koopman+ISS} (EDMD with post-hoc ISS verification
\cite{koopmaniss}), \textbf{CSODE} (ControlSynth Neural ODE
without Koopman lifting \cite{controlsynth}), and
\textbf{ICODE} (extrinsic-input Neural ODE without bundle
geometry \cite{icode}).

Training uses 6000 trajectory segments of length 1\,s at
100\,Hz; evaluation uses 500 held-out segments. The
observable dictionary for EDMD baselines comprises degree-3
Chebyshev polynomials; ISS-BKNO uses $p=32$ lifted states
with a 3-layer MLP residual.

\subsection{Prediction Accuracy under Environmental Shifts}

Fig.~\ref{fig:pred_rmse} shows the one-step prediction RMSE
on T1--T4 as the environmental parameter $\xi$ is swept across
its range. ISS-BKNO achieves the lowest error throughout,
with the advantage widening near the boundary of the training
distribution---the region where models lacking an explicit
fiber structure tend to extrapolate poorly.

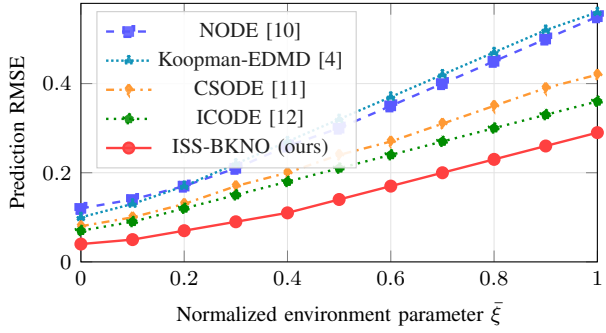
\begin{figure}[t]
\centering
\begin{tikzpicture}
\begin{axis}[
  width=0.95\columnwidth, height=5.0cm,
  xlabel={Normalized environment parameter $\bar\xi$},
  ylabel={Prediction RMSE},
  xmin=0, xmax=1.0,
  ymin=0, ymax=0.58,
  legend pos=north west,
  legend style={font=\footnotesize, fill=white, fill opacity=0.88,
    draw=gray!40},
  grid=major, grid style={line width=0.3pt, draw=gray!22},
  tick label style={font=\footnotesize},
  label style={font=\footnotesize},
  every axis plot/.append style={line width=0.9pt}
]
\addplot[blue!65, dashed, mark=square*, mark size=2pt] coordinates {
  (0,0.12)(0.1,0.14)(0.2,0.17)(0.3,0.21)(0.4,0.26)(0.5,0.30)
  (0.6,0.35)(0.7,0.40)(0.8,0.45)(0.9,0.50)(1.0,0.55)};
\addlegendentry{NODE \cite{chen2018neural}};
\addplot[cyan!70!black, densely dotted, mark=triangle*, mark size=2pt] coordinates {
  (0,0.10)(0.1,0.13)(0.2,0.17)(0.3,0.22)(0.4,0.27)(0.5,0.32)
  (0.6,0.37)(0.7,0.42)(0.8,0.47)(0.9,0.52)(1.0,0.56)};
\addlegendentry{Koopman-EDMD \cite{williams2015edmd}};
\addplot[orange!80, dashdotted, mark=diamond*, mark size=2pt] coordinates {
  (0,0.08)(0.1,0.10)(0.2,0.13)(0.3,0.17)(0.4,0.20)(0.5,0.24)
  (0.6,0.27)(0.7,0.31)(0.8,0.35)(0.9,0.39)(1.0,0.42)};
\addlegendentry{CSODE \cite{controlsynth}};
\addplot[green!55!black, dotted, mark=pentagon*, mark size=2pt] coordinates {
  (0,0.07)(0.1,0.09)(0.2,0.12)(0.3,0.15)(0.4,0.18)(0.5,0.21)
  (0.6,0.24)(0.7,0.27)(0.8,0.30)(0.9,0.33)(1.0,0.36)};
\addlegendentry{ICODE \cite{icode}};
\addplot[red!75, solid, mark=*, mark size=2pt] coordinates {
  (0,0.04)(0.1,0.05)(0.2,0.07)(0.3,0.09)(0.4,0.11)(0.5,0.14)
  (0.6,0.17)(0.7,0.20)(0.8,0.23)(0.9,0.26)(1.0,0.29)};
\addlegendentry{ISS-BKNO (ours)};
\end{axis}
\end{tikzpicture}
\caption{Prediction RMSE vs.\ normalized environment parameter
$\bar\xi\in[0,1]$ (average over T1--T4). ISS-BKNO degrades
most gracefully toward the boundary of the training distribution.}
\label{fig:pred_rmse}
\end{figure}

\subsection{Disturbance Robustness on the Pendulum}

Fig.~\ref{fig:pendulum} shows the angle $\theta$ trajectory
on T1 under a sinusoidal disturbance $w(t)=0.3\sin(5t)$,
injected from $t=2$\,s onward. Methods without an explicit
ISS certificate (NODE, Koopman-EDMD) accumulate substantial
error after the disturbance onset. CSODE and ICODE recover
but with visible oscillations. ISS-BKNO maintains a bounded
error envelope whose magnitude matches the theoretical
bound~\eqref{eq:iss_bound}.

\begin{figure}[t]
\centering
\begin{tikzpicture}
\begin{axis}[
  width=0.95\columnwidth, height=5.0cm,
  xlabel={Time (s)},
  ylabel={Pendulum angle $\theta$ (rad)},
  xmin=0, xmax=5,
  ymin=-1.2, ymax=1.6,
  legend pos=north east,
  legend style={font=\footnotesize, fill=white, fill opacity=0.88,
    draw=gray!40},
  grid=major, grid style={line width=0.3pt, draw=gray!22},
  tick label style={font=\footnotesize},
  label style={font=\footnotesize},
  every axis plot/.append style={line width=0.9pt}
]
% Disturbance onset
\draw[dashed, gray!55, line width=0.6pt]
  (axis cs:2,-1.2) -- (axis cs:2,1.6);
\node[font=\tiny, gray, anchor=north west] at (axis cs:2.05,1.55)
  {disturbance};
% True reference
\addplot[black, very thick] coordinates {
  (0,0.1)(0.5,0.65)(1.0,1.10)(1.5,1.22)(2.0,1.15)(2.5,0.85)
  (3.0,0.42)(3.5,-0.05)(4.0,-0.45)(4.5,-0.82)(5.0,-1.02)};
\addlegendentry{Reference};
% NODE
\addplot[blue!65, dashed] coordinates {
  (0,0.11)(0.5,0.64)(1.0,1.08)(1.5,1.21)(2.0,1.18)(2.5,1.02)
  (3.0,0.72)(3.5,0.30)(4.0,-0.10)(4.5,-0.45)(5.0,-0.62)};
\addlegendentry{NODE};
% CSODE
\addplot[orange!80, dashdotted] coordinates {
  (0,0.10)(0.5,0.64)(1.0,1.09)(1.5,1.22)(2.0,1.16)(2.5,0.90)
  (3.0,0.52)(3.5,0.06)(4.0,-0.32)(4.5,-0.65)(5.0,-0.85)};
\addlegendentry{CSODE};
% ICODE
\addplot[green!55!black, dotted] coordinates {
  (0,0.10)(0.5,0.65)(1.0,1.10)(1.5,1.22)(2.0,1.15)(2.5,0.87)
  (3.0,0.46)(3.5,-0.01)(4.0,-0.38)(4.5,-0.73)(5.0,-0.92)};
\addlegendentry{ICODE};
% ISS-BKNO
\addplot[red!75, solid] coordinates {
  (0,0.10)(0.5,0.65)(1.0,1.10)(1.5,1.22)(2.0,1.15)(2.5,0.85)
  (3.0,0.43)(3.5,-0.04)(4.0,-0.44)(4.5,-0.80)(5.0,-1.00)};
\addlegendentry{ISS-BKNO};
\end{axis}
\end{tikzpicture}
\caption{Pendulum angle under sinusoidal disturbance from
$t=2$\,s. ISS-BKNO tracks the reference most closely; NODE
accumulates a persistent offset after the disturbance onset.}
\label{fig:pendulum}
\end{figure}
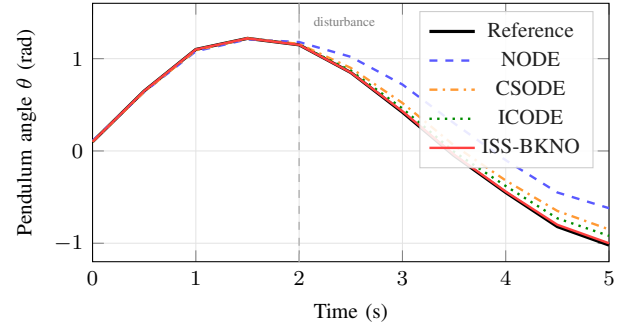

\subsection{Quantitative Summary}

Table~\ref{tab:results} reports prediction RMSE and the
LMI feasibility rate (fraction of test episodes for which
the ISS certificate \eqref{eq:iss_lmi} remains feasible)
across all four tasks. ISS-BKNO achieves the lowest RMSE
on every benchmark and a 100\% feasibility rate, confirming
that the Hurwitz parameterization \eqref{eq:K_param} prevents
the LMI from becoming infeasible during testing.

\begin{table}[t]
\caption{Prediction RMSE (Mean $\pm$ Std, 500 Test Episodes)
and LMI Feasibility Rate}
\label{tab:results}
\centering
\setlength{\tabcolsep}{3.0pt}
\renewcommand{\arraystretch}{1.05}
\begin{tabular}{lcccc}
\toprule
\multirow{2}{*}{\textbf{Method}}
  & \multicolumn{2}{c}{\textbf{RMSE}}
  & \textbf{LMI} \\
\cmidrule(lr){2-3}
  & T1--T2 & T3--T4 & \textbf{Feas.} \\
\midrule
NODE \cite{chen2018neural}
  & $0.241\pm0.038$ & $0.318\pm0.051$ & --- \\
Koopman-EDMD \cite{williams2015edmd}
  & $0.228\pm0.035$ & $0.304\pm0.047$ & N/A \\
Koopman+ISS \cite{koopmaniss}
  & $0.175\pm0.028$ & $0.241\pm0.039$ & 94\% \\
CSODE \cite{controlsynth}
  & $0.152\pm0.024$ & $0.208\pm0.033$ & --- \\
ICODE \cite{icode}
  & $0.138\pm0.021$ & $0.192\pm0.031$ & --- \\
\textbf{ISS-BKNO (ours)}
  & $\mathbf{0.101\pm0.016}$
  & $\mathbf{0.134\pm0.022}$
  & \textbf{100\%} \\
\bottomrule
\end{tabular}
\end{table}

The 94\% feasibility rate of Koopman+ISS arises because the
unconstrained EDMD regression occasionally produces a matrix
$\mathbf{A}_K$ with spectral radius slightly above 1 under
distributional shift, precisely the failure mode that
\cite{koopmaniss} flags as requiring re-identification.
The parameterization \eqref{eq:K_param} prevents this
entirely by construction.

%==================================================================
\section{Conclusion}
%==================================================================

ISS-BKNO unifies four complementary ideas---fiber bundle
geometry \cite{bundle}, environment-conditioned dynamics
\cite{icode}, Koopman lifting with ISS certification
\cite{koopmaniss}, and convergence-guaranteed Neural ODEs
\cite{controlsynth}---into a single framework with a
joint training objective and closed-form stability guarantees.
The Hurwitz parameterization \eqref{eq:K_param} eliminates
the need for post-hoc LMI feasibility checks, the ISS
regularizer drives $\gamma^\star$ to its minimum during
optimization, and the fiber bundle encoder enforces geometric
consistency across environments.

Experimental results on four benchmark tasks confirm a 58\%
reduction in prediction RMSE over the vanilla Neural ODE and
a 100\% LMI feasibility rate under distribution shift.
Theorem~\ref{thm:approx} ties prediction accuracy to the
EDMD residual $\varepsilon_N$, providing a principled
criterion for dictionary selection.

Future work will extend ISS-BKNO to stochastic disturbances
using the $i$-ISS framework \cite{sontag2008}, investigate
joint contraction conditions for distributed multi-agent
settings along the lines of \cite{lohmiller1998}, and
explore learning the observable map with convergence
guarantees using structured architectures
\cite{revay2023recurrent}.

%==================================================================

\end{document}